\documentclass{article}

\usepackage{graphicx}
\usepackage{amssymb, amsmath, amsthm}   
\usepackage[round]{natbib}
\usepackage{textcomp} 
\newtheorem{thm}{Theorem}

\newtheorem{lem}[thm]{Lemma}

\newtheorem{prop}[thm]{Proposition}
\newtheorem{defn}[thm]{Definition}

\newtheorem{prob}[thm]{Problem}
\newtheorem{alg}[thm]{Algorithm}

\newcommand{\eat}[1]{}
\newcommand{\leaves}{\mathcal{L}}
\newcommand{\edges}{\mathcal{E}}
\newcommand{\vertices}{\mathcal{V}}
\newcommand{\rsplits}{\Sigma}
\newcommand{\ersplits}{\Sigma_E}
\newcommand{\srsplits}{\Sigma_S}
\newcommand{\restr}[1]{|_{#1}}
\newcommand{\join}{\wedge}
\newcommand{\coalof}{\succeq}
\newcommand{\strictcoalof}{\succ}

\newcommand{\infimum}{\inf}

\newcommand{\pari}{(i)}
\newcommand{\parii}{(ii)}
\newcommand{\out}{\rho}
\newcommand{\starif}{\bigstar}
\newcommand{\sep}{|}
\newcommand{\fampath}{\mathsf{p}}
\newcommand{\drSPR}{d_{\operatorname{rSPR}}}
\newcommand{\constNJ}{\texttt{constNJ}}
\newcommand{\ocaml}{\texttt{ocaml}}
\newcommand{\tl}{\ell} 
\newcommand{\agr}{\alpha}
\newcommand{\BG}{CRF14\_BG}
\newcommand{\lengths}{\texttt{.lengths}}
\newcommand{\tre}{\texttt{.tre}}

\newcommand{\bF}{\ensuremath{\mathbf{F}}}
\newcommand{\bH}{\ensuremath{\mathbf{H}}}

\newcommand{\rqtilde}{\ensuremath{\hspace{-1pt}_{\mathbf{\tilde{\;}}}}}

\newcommand{\arxiv}[1]{#1}
\newcommand{\noarxiv}[1]{}

\renewcommand{\labelenumi}{(\roman{enumi})}

\title{\constNJ: an algorithm to reconstruct sets of phylogenetic trees satisfying pairwise topological constraints}
\author{Frederick A. Matsen\\
UC Berkeley Dept. Statistics\\
367 Evans Hall \#429\\
Berkeley, CA 94720-3860\\
USA\\
phone: +1 510 642 2450\\
fax: +1 510 642 7892\\
matsen@berkeley.edu\\
http://www.stat.berkeley.edu/\rqtilde matsen/
}

\begin{document}

\maketitle

\noarxiv{\newpage}


\begin{abstract}
  This paper introduces \constNJ, the first algorithm for phylogenetic reconstruction of sets of trees with constrained pairwise rooted subtree-prune regraft (rSPR) distance. 
  We are motivated by the problem of constructing sets of trees which must fit into a recombination, hybridization, or similar network.
  Rather than first finding a set of trees which are optimal according to a phylogenetic criterion (e.g. likelihood or parsimony) and then attempting to fit them into a network, \constNJ\ estimates the trees \emph{while} enforcing specified rSPR distance constraints.
  The primary input for \constNJ\ is a collection of distance matrices derived from sequence blocks which are assumed to have evolved in a tree-like manner, such as blocks of an alignment which do not contain any recombination breakpoints.
  The other input is a set of rSPR constraints for any set of pairs of trees.
  \constNJ\ is consistent and a strict generalization of the neighbor-joining algorithm; it uses the new notion of ``maximum agreement partitions'' to assure that the resulting trees satisfy the given rSPR distance constraints.
\end{abstract}

\noarxiv{\newpage}
\section{Introduction}

Since the pioneering paper of \citet{sneathReticulate75}, tens of thousands of papers have been published on the subject of ``reticulate evolution.'' 
``Reticulate evolution'' has generally come to mean evolution where genetic material for a new lineage may come from two or more sources, as in the case of recombination and hybridization.
The Oxford English Dictionary (1989) defines ``reticulated'' to mean ``constructed or arranged like a net; made or marked so as to resemble a net or network.'' \nocite{reticulatedOED} 
Correspondingly, rather than evolutionary history being representable as a tree, a network is more appropriate.
A considerable amount of effort has gone into the phylogenetic reconstruction of these networks.

Algorithms for phylogenetics in the presence of reticulation have followed a curiously different path then the mainstream of phylogenetics. 
As surveyed below, current algorithms fall into three types: first, there are algorithms which attempt to find the phylogenetic network displaying some fixed characteristics (such as splits in an alignment or some set of trees) which contain the minimum number of reticulation events.
Secondly, there are algorithms to construct ``splits networks,'' which do an excellent job of representing conflicting signals in the data, but do not give an explicit evolutionary history.
The third approach is to sample from the posterior distribution of a population-genetics model, such as the coalescent with recombination.
None of these approaches furnish a practical solution for certain cases, such as HIV researchers who would like to reconstruct the evolutionary history of an alignment which includes recombinant sequences.
Indeed, first fixing a set of characteristics and then minimizing the number of reticulation events ignores the balance between number of reticulation events and phylogenetic optimality,
the splits network approach does not tell a complete evolutionary story,
and population-genetic algorithms are are not yet sufficiently fast for DNA sequence datasets which have thousands of nucleotides.
As there are no algorithms which are practical for doing phylogenetic reconstruction in this setting, HIV researchers who wish to reconstruct evolutionary history typically proceed in one of two ways: they either treat the whole alignment as having a single tree-like history, which cannot possibly be correct, or they build trees on sub-alignments independently, which does not take into account the underlying network structure.
These two extremes, of assuming all trees have the same topology or allowing their topologies to differ in arbitrary ways, leave a substantial gap in the middle, where the correct balance of optimality and discord should be found.

The goal of \constNJ\ is to begin filling this gap in a manner analogous to classical phylogenetic inference algorithms.
To do so, we make a different set of assumptions than has previously been done considering the input and desired output.
Regarding the data, we assume that that the given alignment has been segmented into ``alignment blocks'', each of which can be described by a single tree.
For example, in the case of recombination, the alignment blocks are the segments of an alignment which do not contain recombination breakpoints.
(Note that for the purposes of this paper we will be using the word ``recombination'' in the general sense, including processes such as gene conversion.)
Although the assumption that the data comes pre-segmented is a substantial one, we don't think that it is unreasonable.
From a practical standpoint, some assumption needs to be made, as algorithms which attempt to find a correct segmentation of the data and a sequence of trees simultaneously have a difficult time searching the complete space.
Furthermore, sometimes a segmentation is clear, such as the distinct RNA strands of the influenza genome.
Other times, such as for recombination, it is not so clear, but even in this more difficult case the inference of recombination breakpoints has seen significant progress in the last 10 years (reviewed below).
We will also assume that an outgroup has been selected. 
Such a choice is crucial, as it establishes directionality for reticulation events.

\begin{figure}
  \begin{center}
    \arxiv{\includegraphics[width=2.5in]{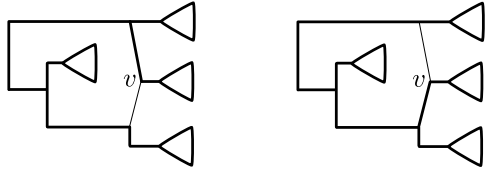}}
  \end{center}
  \caption{
  An example ``reticulate'' network and the two trees that it contains. Those two trees are related by a single rooted subtree prune regraft (rSPR) move, whereby the middle subtree is cut off of the tree and reattached at another location. The node $v$ will be called a reticulation node.}
  \label{fig:reticulate}
\end{figure}

Regarding desired output, rather than actually building a single reticulate network, this paper will focus on building ``correlated sets of trees,'' which display the sorts of constraints found on trees which fit in a reticulate network.
We are focused on building trees because each alignment block is correctly described by a single tree.
However, this set of trees must fit into a network, which forces constraints on their topology.
Specifically, the trees which sit in these networks must be related by rooted subtree-prune-regraft (rSPR) moves, whereby a rooted subtree is cut from the original tree and then re-attached in another location (Figure~\ref{fig:reticulate}).
We describe below how it is necessary for trees sitting in a reticulate network to be related by rSPR moves, though this is not a sufficient condition.

For \constNJ, we assume that the user can supply a series of constraints describing the number of rSPR moves allowable between pairs of alignment blocks.
For example, if the alignment contains ``pure'' types and a single class of recombinants which are derived from a pair of types, then there should be two alignment blocks and the trees for those blocks should be related by one rSPR move as in Figure~\ref{fig:reticulate}.
The challenge, then, is to reconstruct a set of trees which satisfy the constraints and which together optimize some phylogenetically relevant criterion, such as likelihood, parsimony, or balanced minimum evolution.
Note that \constNJ\ actually constructs a number of such sets of trees, in order to display the balance between optimality of the individual trees and the number of reticulation events needed to fit the trees together into a network.

We now present a motivating example.
The \BG\ circulating recombinant form (CRF) of HIV is known to be a mosaic of subtype B and subtype G viruses, and the breakpoints of the recombination events are known \citep{thomsonEaGalicia01}.
We will call the region of the alignment where BG derives from the G subtype the ``G region,'' and the region where BG derives from the B subtype the ``B region.''
As in \citet{thomsonEaGalicia01} and all similar papers we could find in the area, researchers build trees independently on the no-recombination blocks.
We have repeated such an analysis in Figure~\ref{fig:motivIndep}, building PHYML maximum likelihood phylogenetic trees using the F84 model and rooted using CPZ.CD.90.ANT.U42720 (removed from tree for clarity).
As one would hope, the trees do indeed show that the BG CRF derives one portion of its RNA from the G subtype, and the other part from the B subtype.
However, there are many more differences between the two trees than should occur for an alignment with a single recombinant strain. 
For example, the rooting changes between the two trees, as does the location of the C and the F-K clades.
Building a recombination network out of these trees would lead a number of spurious hypothesized recombinations.

\begin{figure}
  \begin{center}
    \arxiv{\includegraphics[width=5in]{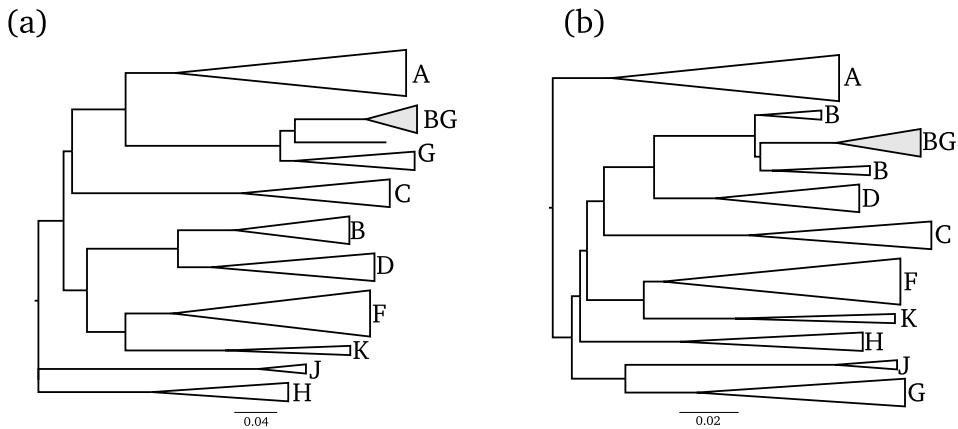}}
  \end{center}
  \caption{
  Phylogenetic trees of the pure subtypes of HIV and the BG recombinant clade constructed independently using the no-recombination blocks of the HIV genome. The single letters (e.g. A,B,C\dots) label clades of subtypes, and BG denotes a clade of circulating recombinant forms (CRFs) made from B and G subtypes. Tree (a) is built from the ``G region,'' i.e. the region where the BG CRF derives from the G subtype, and tree (b) is built from the ``B region,'' where BG derives from the B subtype. As noted in the text, although these trees do place the recombinant strains in the correct locations, they differ in a number of important ways which are not explained by recombination events. It is the perspective of this paper that these extra differences represent phylogenetic error, and that accuracy can be improved by constraining the trees to fit into a recombination network.
  }
  \label{fig:motivIndep}
\end{figure}

In contrast, for this dataset \constNJ\ returns a collection of pairs of trees displaying the balance between the number of allowed rSPR moves between pairs of trees and phylogenetic optimality.
This balance is described in the text output of \constNJ, which is shown in Table~\ref{tab:motiv} for the BG dataset. 
The first column shows the rSPR distance between the two reconstructed trees (in this case the G region tree and the B region tree).
As described below, the notion of optimality for \constNJ\ is total tree length, which is a trivial generalization of the balanced minimum evolution (BME) criterion \citep{desperGascuelBME04}.
It is displayed in the second column for the pairs of trees.
Thus the second line states that \constNJ\ found a pair of trees which differ by a single rSPR move, and which have total tree length about $7.119$.
The third column just shows the difference between the second column values between rows.
Thus $0.0942$ signifies that there is a decrease of magnitude $0.0942$ in total tree length by allowing a single rSPR difference between the two trees.

In this way we can achieve an understanding of the balance between phylogenetic optimality and number of recombination events.
For example, we can see that the decrease in allowing a single recombination event is significantly greater in magnitude than that for allowing two rather than one. 
And surprisingly, allowing nine rSPR moves does not significantly decrease the total tree length compared to allowing three.
Because the improvement in total tree length when allowing one rSPR move is significantly greater than that for any subsequent rSPR moves, we believe that Table~\ref{tab:motiv} suggests that the data probably arose from one recombination event, which agrees with the established knowledge concerning these taxa.

\arxiv{
\begin{table}[ht]
  \centering
  \begin{tabular}{c|c|c}
    rSPR distance & total tree length & tree length difference \\
    \hline
    0                   & 7.213 & 0.0942 \\
    1                   & 7.119 & 0.0076 \\
    2                   & 7.111 & 0.0149 \\
    3                   & 7.096 & 0.0139 \\
    9 (independent NJ)  & 7.082 & \\
  \end{tabular}
  \caption{
  The balance between discord and optimality for the example HIV dataset.
  On the left side is the number of SPR moves required to go from the tree built on the G region to the tree built on the B region.
  In the center is the total tree length (see Equation~\ref{eq:totTreeLength}), which is our notion of optimality.
  On the right is the difference of the total tree length between the rows.
  As described in the text, the largest drop in total tree length comes when allowing a single rSPR move (i.e. recombination event) between the two trees, indicating that one recombination event is needed to explain the data.
  }
  \label{tab:motiv}
\end{table}

}

Furthermore, the trees which \constNJ\ finds assuming a single recombination event agree with the accepted recombination history of the BG recombinant circulating form (Figure~\ref{fig:motivNonIndep}). 
In particular, the only difference between them is the location of the BG clade, which switches from the G to the B subclade depending on the region analyzed.
Importantly, these two trees can fit into a recombination network with a single reticulation node, in contrast to those in Figure~\ref{fig:motivIndep}.

\begin{figure}
  \begin{center}
    \arxiv{\includegraphics[width=5in]{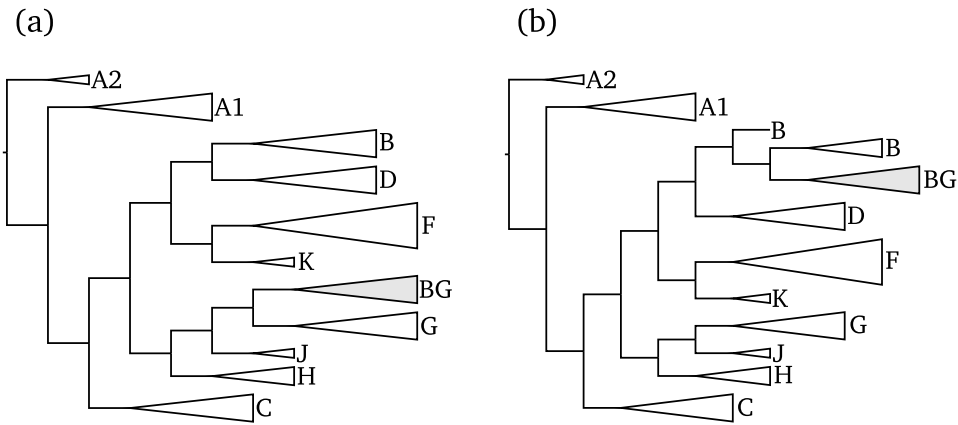}}
  \end{center}
  \caption{
  One pair of phylogenetic trees constructed for the same dataset using \constNJ. 
  In contrast to Figure~\ref{fig:motivIndep}, the only difference between the two trees is the location of the BG recombinant clade.
  These two trees fit into a recombination network with a single recombination event, as should be the case for a tree for the pure subtypes with a single recombinant strain like we have here.
  \constNJ\ correctly identifies that the BG subtype is a recombinant of the B and G subtypes.
  }
  \label{fig:motivNonIndep}
\end{figure}

We believe that \constNJ\ is the first algorithm of its kind, but will now review literature on related topics, starting with common terminology.
The currently accepted term for the class of networks including both hybridization and recombination networks is ``reticulate network''.\footnote{This terminology is redundant, as the word ``reticulate'' already means network-like.}
If we consider a rooted tree to be a directed graph such that edges are directed away from the root, a reticulate network is a rooted phylogenetic tree with additional directed edges making a directed acyclic graph with ``tree nodes'' of in-degree one and ``reticulation nodes'' of in-degree two \citep{husonEaGeneTreeNetwork05}.

A considerable amount of work has gone into the problem of constructing a reticulate network given a set of phylogenetic trees which it must contain. 
This problem was initiated by \citet{maddisonGeneTreeSpeciesTrees97} and considerable progress has been made by \citet{baroniEaReticulateFramework04}, \citet{nakhlehEaReticulateTheoryPractice05}, \citet{husonEaGeneTreeNetwork05}, and \citet{bordewichSempleHybridFPT07, bordewichSempleMinHybrid07}.
As described above, we differ from these approaches as we would like to estimate the trees while ensuring that they fit into a reticulate network.

A related problem (which was the original motivation for fitting trees into a network) is to reconcile distinct gene trees into a single species tree.
This problem has received an appropriately large amount of attention, and has found a more realistic model-based formulation in \citet{aneEaBayesGeneTree07} and \citet{edwardsEaBEST07}.
These differ from the present paper because they assume that there is a single species tree, and that ``correctness'' of a gene tree should in part be judged by the degree to which it fits within a species tree due to a coalescent model.
In our setting, however, there is no single species tree, and the coalescent model may not be appropriate.

Sometimes a related assumption is made, which is not that complete species trees are known, but that the resulting recombination network must display a specified collection of bipartitions, which are typically called splits.
This is equivalent to assuming a supplied alignment evolves according to the infinite sites model of mutation.
The problem again is to find a network which minimizes the number recombination events.
This problem was first formulated by \citet{hudsonKaplanStatRecomb85}, and was shown to be NP-hard in \citet{wangEaPerfectPNetsWRecomb01}. 
Progress was made in a sub-case by \citet{gusfieldEaPhyloNets04} and a simpler related (and in some ways more realistic) problem was solved by \citet{songHeinMinimalARG05}.
In \citet{husonKloepperRecombNet05}, the authors note that the algorithm in \citet{husonEaGeneTreeNetwork05} can be extended to this case.
Although a different formulation, this splits/infinite sites approach represents a different version of the same strategy: find the network displaying a certain set of characteristics which minimizes the number of reticulation events.

Splits network methods are a biologically useful and mathematically interesting way of understanding conflicting signals in phylogenetic data.
The first method to construct splits networks from distance data was the split decomposition approach of \citet{bandeltDressCanonicalDecomp92, bandeltDressSplitDecomp92}.
Another successful approach has been the ``neighbor-net'' algorithm created by \citet{bryantMoultonNeighborNet04} and further analyzed by \cite{levyPachterNeighborNet07}.
These methods form a useful complement to phylogenetic analysis in the traditional tree-based sense, but do not reconstruct an explicit evolutionary history. 
We also note that recombination networks need not be circular split systems, which are the sorts of splits networks returned by neighbor-net.
 
On the other end of the spectrum lie likelihood-based methods using the coalescent with recombination \citep{hudsonCoalRecomb83}.
Major recent advances have been made in this area. 
The full likelihood is quite daunting to compute, but \citet{lyngsoEaRecombParsimony08} have a parsimony-based approach which saves on computation by several orders of magnitude.
Importance sampling \citep{griffithsEaImportanceSampling08} is also promising, but is not yet efficient enough for the long alignments typically encountered in phylogenetics.
Also, it is the intent of this paper to construct a method which is independent of population genetics models such as the coalescent.

A related though distinct line of research is the inference of recombination breakpoints.
One of most basic and most commonly used methods for the inference of recombination breakpoints is called ``bootscanning'', whereby a window is scanned along the alignment and a phylogenetic tree is built for each position of the window; a change in topology between sections of the window can be interpreted as evidence for a recombination breakpoint \citep{loleEaBootscanning99}.
There are many different variations on this theme.
One promising line of research by Marc Suchard and collaborators apply multiple change-point models and reversible-jump MCMC to estimate trees and model parameters along the alignment \citep{suchardEaMCP03, mininEaDualBrothers05}.
We also note that sometimes recombination breakpoints can be seen ``with the naked eye'' as in \citet{thomsonEaGalicia01}.
In contrast to our paper, it is not the intent of these methods to accurately infer phylogeny; furthermore they do not posit any relationship between trees in neighboring no-recombination blocks.
Furthermore, some of the more computationally intensive methods actually require a fixed reference tree.

In summary, we are not aware of any available method for building reticulate networks which 
gives an explicit rooted phylogenetic history for each column of the alignment,
which elucidates the balance of discord between the trees and optimality for those trees,
and which is efficient enough to be useful for modern data sets.
The lack of practical phylogenetic algorithms in the presence of recombination was recently demonstrated in a simulation study by \citet{woolleyEaRecombinationShootoutPlos08}. 
\citet{hugginsYoshidaCophylogeny08} have noted the lack of useful reconstruction algorithms for host-parasite relationships and have noted the need for an algorithm which balances tree concordance and optimality as \constNJ\ does.
Although far from a complete solution for these cases, we believe that \constNJ\ is a first step in the right direction.

\noarxiv{\newpage}
\section{General description of \constNJ}

The primary input for \constNJ\ is a collection of alignment blocks, which as described are disjoint subsets of columns of the alignment which are assumed to evolve in a tree-like manner.
In the case of alignments with recombinant sequences, the alignment blocks are simply the no-recombination blocks.
Note that the alignment blocks need not be contiguous; for example a single recombination event with two recombination breakpoints will result in two, not three, alignment blocks. 
The other input for \constNJ\ is a sequence of constraints on the rSPR distance between the trees constructed for the no-recombination blocks as described below.
Given this input, the goal of \constNJ\ is to exhibit the balance between discordance among the alignment-block trees on one hand, and optimality of the trees in some phylogenetic sense on the other.

\constNJ\ is a deterministic distance-based approach to reconstruction; we chose this direction for several reasons. 
First, the underlying space for a likelihood optimization scheme is even larger than usual, making a heuristic search even less appealing: there are $[(2n-3)!!]^k$ $k$-tuples of rooted bifurcating phylogenetic trees on $n$ taxa. 
The sorts of constraints we will be imposing reduces this number substantially, but little is known about the resulting graph under the sorts of moves typically used in heuristic phylogenetic searches.
Furthermore, likelihood-based approaches are substantially improved by starting with a reasonable tree, which in modern applications is typically a distance-based tree.
Thus, even if a likelihood-based approach was the eventual goal, a distance-based approach would be useful as a ``seed'' for the heuristic likelihood search.
Finally, we feel that distance- and likelihood-based algorithms occupy distinct and complementary roles in the world of computational phylogenetics.

Our goal is to design an approach which generalizes the remarkably accurate and hugely popular neighbor-joining algorithm \citep{saitouNeiNJ87}.
Remarkably, it took almost 20 years for the phylogenetics community to learn the objective function of neighbor-joining; during that time it was even claimed that no such objective function existed.
However, it is now known that neighbor-joining greedily optimizes the ``tree length'' $\tl(T, D)$ (defined below in Equation~\ref{eq:treeLength}) for the given distance matrix $D$.
\constNJ\ generalizes this objective function, as it attempts to minimize the total length of all of all $k$ trees (\ref{eq:totTreeLength}) by a combination of greedy steps.

The trees resulting from \constNJ\ are constrained by the user to be some specified number of rooted subtree-prune-regraft (rSPR) moves from one to another.
As displayed in Figure~\ref{fig:reticulate}, reticulation events such as recombination and hybridization correspond to rSPR tree rearrangements.
The converse is not true: arbitrary rSPR tree rearrangement events need not correspond to reticulation events.
For recombination or hybridization to take place, the participants in the event need to exist at the same time; it is not hard to set up examples of rSPR move combinations which violate this fact (see, e.g., Song and Hein, 2005) \nocite{songHeinMinimalARG05}.
Methods have been developed which take timing restrictions into account \citep{songHeinMinimalARG05, bordewichSempleHybridFPT07}, but we do not incorporate these ideas into a phylogenetic reconstruction framework.
This may be an interesting avenue to for future research, but on the other hand seeing such timing violations can actually be informative.
First, there may be something wrong with the data.
Second, it has been noted \citep{baroniEaRealTime06} that reticulation networks can appear to violate timing constraints if certain taxa are not sampled. The problem of determining the minimal number of ``missing'' taxa required to explain timing constraints has been analyzed by \citet{linzEaAddTaxa09}.
Therefore we have left interpretation of timing issues up to the user of the program.

We now make a more formal statement of the problem \constNJ\ attempts to solve;
note that a similar formulation was made independently by \citet{hugginsYoshidaCophylogeny08} in the context of host-parasite relationships.
\begin{prob}[rSPR-constrained balanced minimum evolution]
  Given $k$ $n \times n$ distance matrices $D_1,\dots, D_k$ and a symmetric $k \times k$ constraint matrix $C$, 
  find the set of trees $T_1,\dots, T_k$ 
  minimizing $\sum_{i=1}^k \tl(T_i, D_i)$
  such that $\drSPR(T_i, T_j) \leq C_{i,j}$ for each $i$ and $j$.
  \label{prob:main}
\end{prob}

\begin{thm}
  \constNJ\ is a consistent algorithm to solve Problem~\ref{prob:main}.
\label{thm:main}
\end{thm}

For \constNJ\ we proceed in a manner analogous to that for neighbor-joining. 
The neighbor-joining algorithm starts with all taxa connected to a central node, then at every stage, chooses the ``coalescence'' (in other papers, ``amalgamation) of trees which most decreases the value of the total tree length.
We mimic this philosophy by evaluating coalescences based on how they affect the total tree length. 
However, in the end we must come up with a collection of trees $T_1,\dots,T_k$ that satisfy the prescribed rSPR constraints.
This raises the question of how one might bound the rSPR distance of the eventual trees ``ahead of time,'' i.e. before the termination of the coalescence steps. 
For instance, if in the developing trees one has the subtrees $(a,b)$ for the first distance matrix, and $(a,c)$ for the second distance matrix, it is clear that the resulting trees must have rSPR distance at least one between trees $T_1$ and $T_2$. 

The question of how to bound eventual rSPR distance is solved by Theorem~\ref{thm:minEquiv}. 
Specifically, we generalize $m$, the size of the maximum agreement forest \citep{bordewichSempleSPR04} to these partially coalesced trees, which forms a sharp bound.
In short, the $m$ value for a pair of partially coalesced trees $T$ and $S$ is the minimum rSPR distance possible among trees resulting from coalescences of $T$ and $S$;
thus once a pair of partially coalesced trees achieves an $m$ value above the corresponding constraint, we can throw that pair out, as the eventual resolved trees will never satisfy the constraints.

\begin{figure}
  \hspace{-1cm}
  \arxiv{\includegraphics[width=5.5in]{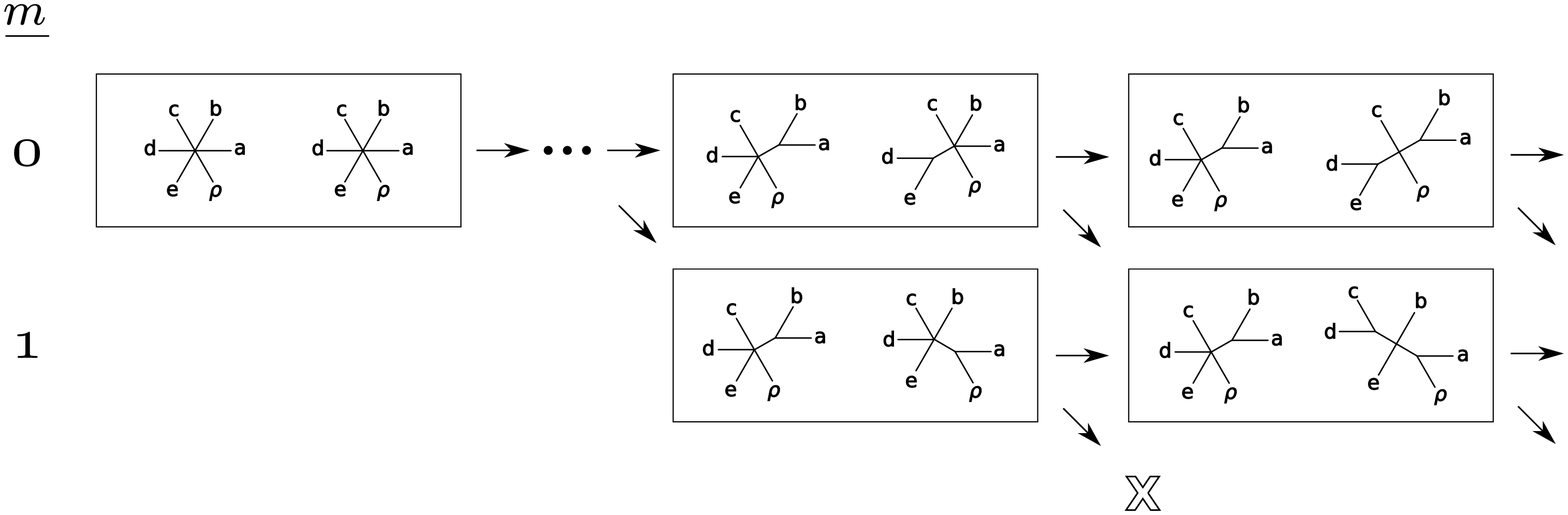}}
  \caption{
  Schematic diagram of the \constNJ\ algorithm. As described in the text, at every stage we attempt to find the optimal pair of partially coalesced trees which could eventually be at most some fixed number of rSPR moves apart. As shown in Theorem~\ref{thm:minEquiv}, the $m$ value for a pair of partially coalesced trees forms a sharp lower bound for the eventual rSPR distance between those trees. Therefore pairs of partially coalesced trees which have $m$ value exceeding the constraint on rSPR distance can be thrown out, as shown by the X.
  }
  \label{fig:algIdea}
\end{figure}

Using this $m$ we construct our greedy algorithm, as shown in Figure~\ref{fig:algIdea}.
Say that we only have two trees, and that we want to find the minimal-total-length pair of trees which are only one rSPR move apart.
At every stage, we attempt to find the best pair of partially coalesced trees with $m$ values zero and one.
We start with two star trees; $m$ applied to this pair is zero.
The first-step coalescence must also lead to a pair of trees which have $m$ value zero, as one of the trees is still completely unresolved.
Say the optimal, in terms of total tree length, second step NJ-type coalescence leads to a pair of trees which have $m$ value one (indicated by the first diagonal arrow in Figure~\ref{fig:algIdea}). 
Then we go down the list of second-step coalescences for the trees, and find the best one which does not increase the $m$ value at all (indicated by a horizontal arrow in Figure~\ref{fig:algIdea}).
Next we repeat the process for each of the trees from the previous stage, saving the best pair of trees which have $m$ values zero and one.
In the end, we will have the best pair of trees which have rSPR distances zero and one which were achievable via a series of greedy steps. 
Although not guaranteed to be the optimal pair of trees, the algorithm is consistent.

\noarxiv{\newpage}
\section{Technical preliminaries}

In this section we review some definitions and clarify our notion of optimality.
As stated in the introduction, we will always assume that an outgroup taxon has been
chosen, and will label it $\out$. Thus 
\textbf{we always assume that} $\out$ \textbf{is contained in any
taxon set} $X$.
We will use the following definitions. For the
purpose of this paper, a \emph{tree} on a finite taxon set $X$ will be
a rooted binary phylogenetic $X$-tree. 
A \emph{forest} on taxon set
$X$ will be a collection of trees on disjoint taxon sets such that the
union of the taxon sets is $X$. 
We will sometimes consider a
tree on $X$ to be a forest with a single tree. 
An \emph{unrooted tree} on a finite taxon set $X$ will be an unrooted
phylogenetic $X$-tree (note that unrooted trees will be allowed to have
multifurcating nodes.) $\leaves(R)$, $\edges(R)$, and $\vertices(R)$
will denote the leaves, edges, and vertices of a tree, unrooted tree,
or forest $R$. 

Although $\out$ represents the true rooting of the phylogenetic tree, \textbf{we will not always assume that our trees or forests are rooted at} $\out$. 
We must do so because the NJ-type coalescences will not in general root the tree at the edge leading to $\rho$. 
Therefore, we must allow alternative rootings, but at the same time keep in mind that the rSPR distance between the trees must be calculated with respect to the edge leading to $\rho$. 
Thus we use the following definition of rSPR on an unrooted tree: given an unrooted tree $U$ on a taxon set $X \ni \out$, a single SPR move first cuts some edge of the tree except for that leading to $\out$, resulting two rooted trees $R$ and $S$. 
Say $\rho \in \leaves(R)$.
Suppress the degree two root node of $R$, and attach $S$ to some edge of the resulting unrooted tree by inserting a degree two node onto the chosen edge, then connecting the root of $S$ to that new node. 
This definition is the same as that of \citet{bordewichSempleSPR04} when considering trees rooted at the edge leading to $\rho$.

As with any distance function defined implicitly in terms of a graph, the minimum number of rSPR moves required to transform one tree $T$ into another $S$ is a metric; we define $\drSPR(T,S)$ to be this number.

\subsection{Tree length and the balanced minimum evolution criterion}

As reviewed by \citet{gascuelSteelNJRevealed06}, phylogenetics researchers now understand the optimality function of the neighbor-joining algorithm \citep{saitouNeiNJ87}.
Let $\fampath(i,j)$ denote the path from $i$ to $j$ in the unrooted tree $T$,
and define the \emph{weight} of a path from leaf $i$ to leaf $j$ as
\[
 w(i,j) = \prod_{v \in \fampath(i,j)} \frac{1}{\deg(v)-1}.
\]
Then the ``length'' of an $n$ taxon tree $T$ with respect to an $n \times n$
distance matrix $D$ is \citep{sempleSteelCyclic04}:
\begin{equation}
  \tl(T,D) = \sum_{i,j} w(i,j) D_{i,j}.
  \label{eq:treeLength}
\end{equation}
The name ``tree length'' comes from the fact that if $D$ is a distance matrix
derived from some assignment of branch lengths to the edges of $T$, then $\tl$
will be the total length of all of the edges. However, the name may be somewhat
confusing initially, because $\tl$ need not be defined as sum of the branch
lengths of any specific tree.

The tree $T$ which minimizes $\tl(T,D)$ for some distance matrix $D$ is known as the \emph{balanced minimum evolution} (BME) tree for the distance matrix $D$. 
The BME criterion is consistent \citep{desperGascuelBME04}, and neighbor-joining is a consistent tree-building heuristic which greedily minimizes total tree length \citep{desperGascuelME05} 
As described in Problem~\ref{prob:main}, \constNJ\ attempts to minimize 
\begin{equation}
  \sum_{i=1}^k \tl(T_i, D_i)
  \label{eq:totTreeLength}
\end{equation}
while enforcing pairwise constraints on the rSPR distance between pairs of trees. 
When $k=1$ \constNJ\ is simply neighbor joining, while for $k>1$ \constNJ\ is a strict generalization of NJ.

\noarxiv{\newpage}
\section{Rooted SPR and maximum agreement partitions}
\label{sec:rSPR-MAP}

This section describes the primary technical content of this paper. 
As described in the introduction, we would like to proceed via coalescences in a manner similar to neighbor-joining, while ensuring that the eventual rSPR distance between the trees is not too large.
In order to assure adherence to the rSPR criterion, we develop the notion of maximum-agreement partition, which generalizes the notion of maximum agreement forest from \citet{bordewichSempleSPR04}. 
As shown in Theorem~\ref{thm:minEquiv}, maximum agreement partitions and the associated $m$ value allow us to bound the rSPR distance between the two partially resolved trees ``in advance.''


\subsection{Compatibility and coalescence}
We will use the following definitions. 
A \emph{split} on a taxon set $X$ is a bipartition of $X$. Because the set $X$
will be clear, we will often abuse notation by identifying $A \subseteq X$ with the 
partition $A | (X \setminus A)$.
Furthermore, because we have a special element $\out$,
we can distinguish between the two sides of a split; the
side not containing $\out$ will be called the \emph{rsplit} (short for
rooted split) of the split. 
It is clearly equivalent to describe a given partition in terms of a split or an rsplit, and we will use the two descriptions interchangeably.

Note that the neighbor-joining algorithm is typically thought of as proceeding by coalescing internal nodes of an unresolved phylogenetic tree (see Figure~\ref{fig:algIdea}); however for our purposes it will sometimes be easier to consider the forest obtained by deleting the central node and the associated edges.
The opposite construction will be called ``starification.''
\begin{defn}
  Given a forest $F$, define the \emph{starification} $\starif(F)$ of $F$ 
  as the following unrooted tree.
  If $F$ has one tree, then suppress the degree two root node of $F$.
  If $F$ has two trees, then join their root nodes by an edge. 
  If $F$ has three or more trees, join all of the root nodes of trees of $F$ to a
  single node. The new introduced node will be called the \emph{star
  node}.
\end{defn}
\noindent
We will identify any one, two, or three tree forest $F$ with its starification, in which case there is no designated star node.

\begin{defn}
  Given a tree $T$ which is part of a forest $F$ on a taxon set $X$,
  define the \emph{edge splits} $\ersplits(T)$ to be $\leaves(T) \sep
  [\leaves(T)]^c$ along with the set of splits on $X$ induced by the
  edges of $T$. We define $\ersplits(F)$ to be the union of the edge splits
  of $T$ across all trees $T$ in $F$.
\end{defn}

For example, the rsplits $\{3\}$ and $\{2,3,4\}$ are both edge rsplits of the forest $((1,\out),2); (3,4)$.

\begin{defn}
  Given a forest $F$ on a taxon set $X$, $A$ is a
  \emph{separating split} of $F$ if $A$ is the union of taxon sets
  for a collection of at least two trees of $F$. The set of separating splits of
  $F$ will be denoted $\srsplits(F)$. 
\end{defn}
\noindent
Given a forest $F$ we will write $\rsplits(F)$ for $\ersplits(F) \cup \srsplits(F)$.
This will be the set of splits used to make agreement partitions as described below.

\begin{defn}
  Two rsplits $A$ and $B$ will be called \emph{compatible} if either
  $A \cap B = \emptyset$, $A \subseteq B$, or $B \subseteq A$. 
\end{defn}
Because $A$ and $B$ are the sides of the splits which
do not contain $\out$, this is the same as the usual criterion for
split compatibility \citep{sempleSteelPhylogenetics03}. Therefore we
have the following well known theorem. 

\begin{thm}[Buneman, 1971]
  \label{thm:splitsEquiv}
  A collection of splits $M$ on a taxon set $X$ is pairwise compatible iff there exists an unrooted tree $T$ on taxa $X$ such that $M$ is a subset of $\ersplits(T)$. 
  There is a one-to-one correspondence between compatible sets of splits on $X$ and minimally-resolved unrooted trees on $X$.
\end{thm}

\begin{defn}
  Two forests $F$ and $G$ on taxon set $X$ are \emph{compatible} if
  $\ersplits(F)$ and $\ersplits(G)$ are pairwise compatible.
\end{defn}

\begin{defn}
  The \emph{join} $T \join S$ of two trees $T$ and $S$ on disjoint taxon sets
  is the tree obtained by joining the root nodes of $T$ and $S$ to a
  new root node. The \emph{coalescence} of $T$ and $S$ in the forest
  $F$ is the forest $\{T \join S\} \cup (F \setminus \{T, S\})\}$.
  \label{defn:joinCoal}
\end{defn}
Note that the operation of coalescence gives a partial order on
the set of forests on a given taxon set. Namely, we write $F \coalof
G$ if $F$ is a coalescence of $G$. Clearly, trees are the maximal
elements in this partial order.

\begin{defn}
  A tree $S$ is a \emph{subtree} of an unrooted tree $U$ if $S$ is one component of the disconnected graph obtained by cutting an edge of $U$. 
  A tree $S$ is a \emph{subtree} of a rooted tree $T$ if $S$ is a component of the disconnected graph obtained by cutting an edge of $T$, and $S$ does not include the root of $T$. 
  \label{defn:subtree}
\end{defn}
We emphasize that the subtree definition is different than than that of an
\emph{induced subtree}, which is as follows. The existence of induced subtrees is
guaranteed by Theorem~\ref{thm:splitsEquiv} or its rooted equivalent.

\begin{defn}
  Given a tree $T$ and $Y \subseteq \leaves(T)$, $T\restr{Y}$ is the
  (rooted or unrooted) tree on taxa $Y$ with rsplits $\{A \cap Y: A \in \ersplits(T)\}$.
\end{defn}

There is also an analogous definition for forests.
\begin{defn}
  Given a forest $F$ and $Y \subseteq \leaves(F)$, $F\restr{Y}$ is 
  \[
  \{ T \restr{Y} : T \in F \hbox{ and } L(T) \cap Y \neq \emptyset \}.
  \]
\end{defn}

\begin{prop}
  If two forests $F$ and $G$ on a taxon set $X$ are compatible, and
  $F$ has more than one tree, then there exists $H \strictcoalof F$
  such that $H$ is compatible with $G$.
  \label{prop:compatCoal}
\end{prop}

\begin{proof}
  If $F$ has two or three trees, the proposition is trivial.
 Otherwise, let $U$ be the tree with split set equal to the union of
 $\ersplits(F)$ and $\ersplits(G)$. If $U$ is not resolved (i.e. if
 there exists an internal node of degree greater than three) then take an 
 arbitrary resolution. As all of the trees $T$ of $F$ are
 resolved, each $T$ sits as a subtree of $U$; let $J$ be the union of
 the nodes of the $T \in F$ (considered as nodes of $U$). 
 Let $\fampath$ denote the longest path in $U$ which
 does not contact any of the nodes in $J$. Because $F$ has at least
 four trees, $\fampath$ will be nontrivial. Pick one end of this path,
 which must be connected to a pair of trees $S', S''$ of $F$. Let $K =
 \leaves(S') \cup \leaves(S'')$. As the split $K \sep K^c$ is already
 a split of $U$, we know that it is compatible with $\ersplits(F)$ and
 $\ersplits(G)$, and thus that $\rsplits(S' \join S'')$ is compatible 
 with $\ersplits(G)$. Let $H$ be the coalescence of $S'$ and $S''$ in $F$. 
\end{proof}

\subsection{Maximum agreement partitions}
In this section we introduce the notion of maximum agreement partition
(MAP), which generalizes the idea of maximum agreement forests. 
Maximum agreement forests were first introduced by 
\citet{heinEaComplexity96}, and further refined by \citet{bordewichSempleSPR04}. 
In broad terms, given two forests $F$ and $G$ on a taxon set $X$, we
will be interested in considering partitions $P$ which are obtainable
from $F$ and $G$ independently by ``combining'' edge splits and
separating splits of those forests, in the same way that edge cuts
are combined when making maximum agreement forests. The
appropriate notion of ``combining'' splits is the infimum, which we now describe. 

The set of partitions on a given finite set $Y$ form a partial order,
such that a partition $P_1 \leq P_2$ if $P_1$ is a refinement of
$P_2$. In fact, the set of partitions is a \emph{complete lattice},
meaning that any set of partitions on $Y$ has a supremum and an
infimum. For a collection of partitions $M$, we will use $\infimum(M)$
to denote their infimum. 

Thus, as described below, a necessary condition for $P$ to be an
agreement partition for two forests $F$ and $G$ is that $P$ can be
expressed as $\infimum(M)$ and $\infimum(N)$ for $M \subseteq
\rsplits(F)$ and $N \subseteq \rsplits(G)$. It will now be
useful to connect that definition to one in terms of convexity of
characters \citep{sempleSteelPhylogenetics03}. 

\begin{defn}
  Given a partition $P$ on some set $K$, define $P \restr{J}$ for some
  $J \subseteq K$ to be the partition $\{Y \cap J : Y \in P\}$.
\end{defn}

The following is a slight generalization of the definition of convexity
given by \citet{sempleSteelPhylogenetics03}. 
\begin{defn}
  A partition $P$ on a taxon set $X$ is convex on a
  forest $F$ on $X$ if there exists an $H \coalof F$ such that
  $P$ induces a convex character on
  $\starif(H)$, i.e. if there exists a partition $\tilde{P}$ on
  vertices $\vertices(\starif(H))$ such that
  \begin{enumerate}
    \item $P = \tilde{P} \restr{X}$.
    \item Any $\tilde{Y} \in \tilde{P}$ separates $\starif(H)$ into connected components.
  \end{enumerate}
  \label{defn:convex}
\end{defn}

The following proposition relates the notions of ``obtainable by a series of
cuts along edge or separating splits'' with the notion of character convexity. 
\begin{prop}
  A partition $P$ of a taxon set $X$ is convex on a forest $F$ iff
  there exists $M \subseteq \rsplits(F)$ such that $P = \infimum(M)$.
  \label{prop:convexInfimum}
\end{prop}
\begin{proof}
  Assume $M \subseteq \rsplits(F)$ such that $P = \infimum(M)$. Note
  that $K = \inf(M \cap \srsplits(F)$ is a set of disjoint separating or root-edge
  splits for $F$; thus we can perform coalescences, making $H$, such
  that the splits from sets in $K$ are edge splits of $H$. Such an
  $H$ will satisfy the criteria of the definition.

  For the converse implication, cutting any edge $(u,v)$ of $\starif(H)$ for
  any $H \coalof F$ gives a split in $\rsplits(F)$. We then define $M$
  as the set of such splits $s_{u,v}$ such that such that $(u,v)$ is 
  an edge and $u$ and $v$ are in distinct sets of the partition $\tilde{P}$.
  By construction, $P = \infimum(M)$.
\end{proof}

The following definition generalizes the notion of agreement forest.
\begin{defn}
  We say that a partition $P$ of taxon set $X$ is an
  \emph{agreement partition} for a pair of forests $F, G$ on $X$ if
  \begin{enumerate}
    \item for every pair of rsplits $A \in \ersplits(F)$, $B \in \ersplits(G)$, and $Y
      \in P$, $A \cap Y$ is compatible with $B \cap Y$.
    \item $P$ is convex on $F$ and $G$.
  \end{enumerate}
  We say that $P$ is a \emph{maximum agreement partition} (MAP) if the
  number of sets of $P$ is less than or equal to that of any other
  agreement partition. Let $m(F,G)$ be the number of sets in the MAP
  minus one. 
  \label{def:MAP}
\end{defn}

Note that by Theorem~\ref{thm:splitsEquiv}, for two resolved unrooted trees
$U,V$ on a taxon set $X \ni \out$, the size of the maximum agreement partition
is the same as the size of the maximum agreement forest of the trees (rooted at
$\out$) in the sense of \citet{bordewichSempleSPR04}. 
Recall that the definitions of maximum agreement forest in \citet{bordewichSempleSPR04} differs from that of \citet{heinEaComplexity96} and \citet{allenSteelSPR01}.

\begin{prop}
  Assume $F$, $G$, and $H$ are forests on a taxon set $X$ such that $H
  \coalof F$, and $P$ is an agreement partition for $H$ and $G$. Then
  $P$ is also an agreement partition for $F$ and $G$.
  \label{prop:coalIncr}
\end{prop}

\begin{proof}
  Part \pari\ of the definition is clear as $\ersplits(F) \subseteq
  \ersplits(H)$. Next we check \parii, i.e. that $P$ is convex on
  $F$. Note that $H \coalof F$ implies $\rsplits(H) \subseteq
  \rsplits(F)$, as the ``extra'' edge splits of $H$ will be separating
  splits of $F$. By Proposition~\ref{prop:convexInfimum}, there exists an $M
  \subseteq \rsplits(H)$ such that $P = \inf(M)$; by the previous sentence $M
  \subseteq \rsplits(F)$ and so by Proposition~\ref{prop:convexInfimum} again $P$
  is convex on $F$. 
\end{proof}

The following theorem is the main motivation for studying
the maximum agreement partition. Thus the proposition says
that the size of the maximum agreement partition of the two forests
$F$ and $G$ is the same as the rSPR distance in the best case. 
\begin{thm}
  The minimum of $\drSPR(U,V)$ across all unrooted trees
  $U \coalof F$ and $V \coalof G$ is equal to $m(F,G)$. 
  \label{thm:minEquiv}
\end{thm}
\noindent The proof of this proposition will come after two lemmas.

\begin{lem}
  Given a partition $P$ convex on $F$ and $Y \in P$ such that $F \restr{Y}$ includes two distinct trees $T$ and $S$, then there exist distinct trees $\tilde{T}, \tilde{S} \in F$ such that $\left. \left( \tilde{T} \join \tilde{S} \right)\right| _{Y} = T \join S$. 
  Furthermore, for any $Z \in P$ not equal to $Y$ and any $R \in \{\tilde{T}, \tilde{S}\}$, we have either $Z \subset \leaves(R)$ or $Z \cap \leaves(R) = \emptyset$.
  \label{lem:joinRestr}
\end{lem}

\begin{proof}
  Let $H$ and $\tilde{P}$ be as in Definition~\ref{defn:convex}. 
  Let $\tilde{Y} \in \tilde{P}$ be such that $\tilde{Y} \cap \leaves(F) = Y$.
  Let $\tilde{T}$ (resp. $\tilde{S} \in F$) be the tree such that $\tilde{T} \restr{Y} = T$ (resp. $\tilde{S} \restr{Y} = S)$.
  Let $Q = \leaves(T) \cup \leaves(S)$. 

  We now show that $\tilde{T} \neq \tilde{S}$. 
  The contrary would imply $\starif(H) \restr{Q} = \tilde{T} \restr Q.$
  Because $Q \subset Y$ and $\starif(H) \restr Y$ is connected by definition, $\tilde T \restr Q$ is connected so $T$ and $S$ would not be distinct.
  This is a contradiction.
  Thus $\leaves(\tilde T) \cap \leaves(\tilde S) = \emptyset$ so $\left. \left( \tilde{T} \join \tilde{S} \right)\right| _{Y} = \tilde T \restr Y \join \tilde S \restr Y = T \join S$.

  We now show the second statement of the lemma.
  Let $r(W)$ denote the root node of any tree $W \in F$.
  Note that $r(\tilde T)$ and $r(\tilde S)$ must be in $\tilde Y$ because $\starif(H) \restr Q$ is connected and the path between any $a \in \leaves(T)$ and $b \in \leaves(S)$ passes through $r(\tilde T)$ and $r(\tilde S)$.

  Now assume that for some $R \in \{\tilde{T}, \tilde{S}\}$ we have that some $Z \neq Y$ of $P$ intersects $\leaves(R)$ but is not contained in it.
  Take $c \in Z \cap \leaves(R)$ and $d \in Z \cap \left[ \leaves(R) \right]^c$.
  Let $\tilde Z \in P$ be such that $\tilde Z \cap L(F) = Z$.
  By the same argument as in the previous paragraph, $r(R)$ is in $\tilde Z$.
  This is a contradiction as $\tilde Y$ and $\tilde Z$ are disjoint.
\end{proof}

\begin{lem}
  Assume that $F$ and $G$ are forests on a taxon set $X$, and $P$ is an
  agreement partition for $F$ and $G$. Then there exist resolved trees $U \coalof F$
  and $V \coalof G$ such that $P$ is an agreement partition for $U$ and $V$.
  \label{lem:commonCoal}
\end{lem}

\begin{proof}
  It is enough to show that if one of the forests, say $F$, has at
  least four trees then there exists an $H_0 \strictcoalof F$ such
  that $P$ is an agreement partition for $H_0$ and $G$. 

  If for every $Y \in P$ we have that $F\restr{Y}$ is a single tree,
  then we can make $H_0$ by taking an arbitrary coalescence of $F$;
  any such coalescence will be ``broken'' by $P$ and thus will not
  introduce any splits violating \pari\ of Definition~\ref{def:MAP}. Thus we assume that $F\restr{Y}$
  has at least two trees. By Proposition~\ref{prop:compatCoal}, there
  exist nontrivial $T, S \in F\restr{Y}$ such that the coalescence of
  $T$ and $S$ in $F\restr{Y}$ is compatible with $G\restr{Y}$. 
  
  By Lemma~\ref{lem:joinRestr}, there exist $\tilde{T}$ and $\tilde{S}$ in $F$ such that
  $\left. \left[ \tilde{T} \join \tilde{S} \right]\right| _{Y} = T \join S$. 
  Let $H_0$ be the coalescence of $\tilde{T}$ and $\tilde{S}$ in $F$;
  the second statement of Lemma~\ref{lem:joinRestr} implies that the coalescence of $\tilde T$ and $\tilde S$ does not introduce any new edge splits when restricted any $Z \neq Y$ in $P$, and so $H_0$ satisfies the criterion \pari\ of a maximum agreement partition.

  Also, $P$ is convex on $H_0$, establishing criterion \parii.
  Indeed, by Proposition~\ref{prop:convexInfimum}, let $M \subseteq \rsplits(F)$ be such that $P = \infimum(M)$; we need to show that $M \subseteq \rsplits(H_0)$. 
  The only difference between $\rsplits(F)$ and $\rsplits(H_0)$ is that $\rsplits(H_0)$ does not have separating partitions which separate $\tilde{T}$ and $\tilde{S}$, but $M$ cannot contain such a partition because $T$ and $S$ both have taxa in the same partition of $P$.
\end{proof}

\begin{proof}[Proof of Proposition~\ref{thm:minEquiv}]
  Lemma~\ref{lem:commonCoal} shows that the minimum of $m(U,V)$ is
  less than or equal to $m(F,G)$. 
  The other inequality follows from Proposition~\ref{prop:coalIncr}.

  Now note that for a resolved tree on $X$ rooted at $\out$, the
  notions of maximum agreement forest and maximum agreement partition
  coincide. Thus by Theorem~2.1 of \citet{bordewichSempleSPR04},
  $m(U,V)$ is equal to the rSPR distance between $U$ and $V$ for any
  resolved $U \coalof F$ and $V \coalof G$.
\end{proof}

\subsection{Calculating the maximum agreement partition}

As introduced above, and described more clearly below, \constNJ\ needs to find a great number of agreement partitions.
Indeed, a sample \constNJ\ run with three distance matrices, 27 taxa, with pairwise constraints of size two required 5867 calls to the subroutine finding the size of a MAP.
Therefore a speedy calculation of the MAP is essential. 

In the present implementation of \constNJ, the MAP is calculated is via a simple extension of the algorithm by \citet{bordewichSempleSPR04}.
As with the usual Bordewich-Semple algorithm, we contract isomorphic subtrees and replace chains of pendant subtrees with chains of three pendant edges. 
However, we consider separating rsplits as well as edge rsplits to find the agreement partition.

An alternative would be to consider an integer linear programming (ILP) approach to the MAP problem based the work of Yufeng Wu, who has recently developed an ILP approach to finding a maximum agreement forest \citep{wuExactSPR08}. 
Although Wu's ILP approach is many orders of magnitude faster than the Bordewich-Semple algorithm for finding the size of the maximum agreement forest in the ``hard'' case when two trees are quite different, our tests have shown that it is slower in the ``easy'' case.
This difference is probably because there is overhead to creating the linear programming matrix, which does not scale strongly with respect to the difficulty of the problem, while the Bordewich-Semple algorithm is very fast for easy problems.
It is possible that some of the ILP overhead could be amortized by clever re-use of portions of the matrix across coalescences, or a combination of Bordewich-Semple and Wu ideas, but we have not followed these directions.

\noarxiv{\newpage}
\section{The \constNJ\ algorithm}

\renewcommand{\labelenumi}{\arabic{enumi}.}
\renewcommand{\labelenumii}{\alph{enumii}.}

Assume \constNJ\ is given $k$ distance matrices on a taxon set $X$.
On the way to constructing our trees $T_1,\dots, T_k$ on $X$ we will be constructing collections of forests $\bF = F_1,\dots, F_k$;
we will call such a collection $\bF$ an ``instance.''
For example, each boxed pair of trees in Figure~\ref{fig:algIdea} is an instance (after deleting the central ``star'' nodes). 
The \emph{agreement profile} for an instance $\bF$ is the $k \times k$ matrix $\agr(\bF)$ where $\agr(\bF)_{ij}$ is $m(F_i, F_j)$.
It describes the degree to which the forests agree.
The \emph{identical agreement profile} is the $k \times k$ zero matrix.
Define the \emph{instance tensor} to be a partially filled tensor of instances indexed by $\mathbb{N}^{k^2}$, where $\bF$ is stored in the ``slot'' indexed by its agreement profile $\alpha(\bF)$.
\begin{alg}[\constNJ]
  Given $n \times n$ distance matrices $D_1, \dots, D_k$ and a $k \times k$ constraint matrix $C$,
  \begin{enumerate}
    \item Let $\bF^{(0)}$ be the trivial instance, i.e. $F^{(0)}_i$ is the trivial forest on $n$ taxa for each $1 \leq i \leq k$.
      Let $\bH^{(0)}$ be the instance tensor containing only $\bF^{(0)}$.
    \item Repeat the following until termination:
      \begin{enumerate}
	\item Let $\bH$ be the instance tensor from the previous step.
	\item Rank all possible coalescences of all of the instances of $\bH$ by how much they will decrease total tree length.
	\item Make a ``step'' by walking down this ranked list in order as follows:
	  \begin{enumerate}
	    \item Perform the chosen coalescence, say of an instance $\bF$, and assume that the resulting instance $\bF'$ has agreement profile $X$.
	    \item If some entry of $X$ is greater than the corresponding element of $C$, discard $\bF'$ and test the next coalescence.
	    \item If not, and $\bF'$ is the first in this step to have agreement profile $X$, then save it. If, on the other hand, another instance has already been found in this step with agreement profile $X$, then discard $\bF'$ as it must have a larger total tree length.
	    \item Stop walking down the list if $X$ is the identical agreement profile.
	  \end{enumerate}
	\item Terminate if each of the $F_i$ have three trees or fewer.
      \end{enumerate}
  \end{enumerate}
  \label{alg:main}
\end{alg}

We now show that this algorithm is consistent.
\begin{proof}[Proof of Theorem~\ref{thm:main}]
  In broad terms, Algorithm~\ref{alg:main} is consistent because of the consistency of neighbor-joining 
  \citep{gascuelConcerningNJ97,bryantNJUniqueness05}
  and because the coalescence which most decreases the total tree length must be a neighbor-joining step \citep{desperGascuelME05}.
   We are given a sequence of distance matrices $D_1,\dots,D_k$ and a symmetric $k \times k$ constraint matrix $C$. 
   By hypothesis, these distance matrices come from a sequence of trees $T_1,\dots,T_k$ such that the rSPR distance between $T_i$ and $T_j$ is bounded above by $C_{i,j}$.
   First, by the consistency of neighbor-joining, NJ applied to each distance matrix independently will recover the correct collection of trees. 
   Say the sequence of neighbor-joining coalescences making $T_i$ gives a series of forests $F_{1,i},\dots,F_{n-1,i}$, where $F_{n-1,i} = T_i$.
   Thus by Theorem~\ref{thm:minEquiv} (more specifically, Proposition~\ref{prop:coalIncr}) and our assumptions about the $T_i$, 
   \begin{equation}
     m(F_{a,i}, F_{b,j}) \leq C_{i,j}
     \label{eq:forestsOK}
   \end{equation}
   for any $1 \leq a,b \leq k$ and $1 \leq i,j \leq n$.
   Thus the constraints will always be satisfied as long as we follow the sequence of NJ steps for each tree.
   
   Next we show by induction that given this data, at every step every \constNJ\ forest will one of the $F_{r, j}$ for $1 \le r < n$ and $1 \le j \le n-1$.
   This is clearly true at initialization.
   By induction, assume the assertion is true at some \constNJ\ step.
   Consider the coalescence which decreases total tree length as much as possible irrespective of constraints; say it occurs in $F_{k,i}$.
   As the coalescence decreases the tree length of $F_{k,i}$ compared to other coalescences of $F_{k,i}$, is also a neighbor-joining step for $F_{k,i}$, making $F_{k+1,i}$.
   By the previous paragraph, we know that this coalescence will preserve the constraints, and thus is also a \constNJ\ step 
   (recall that each \constNJ\ step decreases the total tree length as much as possible amongst coalescences which preserve the constraints).
   Thus at the end we get $F_{n-1,i}$ for each $i$ by induction. 
   Because $F_{n-1,i} = T_i$, \constNJ\ is a consistent algorithm.

\end{proof}

\subsection{Implementation}
\label{sec:implementation}

We have implemented \constNJ\ in the fast functional/imperative language \ocaml\ \citep{ocaml}.
The implementation has a simple command line interface, which is documented in the accompanying manual.
It is available for download from the author's website, at
http://www.stat.berkeley.edu/\rqtilde matsen/constNJ/ .

As described above, the primary input for \constNJ\ is a series of distance matrices, with one for each alignment block.
The program is designed to accept distance matrices from the DNADIST program of the PHYLIP package, although longer lines and taxon names are allowed.
The first taxon is assumed to be the outgroup.
The program assumes that the taxa in the distance matrices are ordered in a corresponding way. 
For instance, if one is using \constNJ\ to investigate recombination, all of the taxa should be listed in the same order, so that the taxa in the no-recombination blocks correspond to one another.
On the other hand, if one is using \constNJ\ to investigate host-parasite relationships, the $i$th taxon in the parasite alignment should parasitize the $i$th taxon in the host alignment.
If, for example, a given parasite is present in multiple hosts, this will require duplication of that parasite sequence in the alignment.

The second input for \constNJ\ is a set of constraints for the resulting correlated set of trees.
There are two options for specifying these constraints: first, via a file, or second, by enforcing ``linear'' constraints.
For example, assume we supply three distance matrices: $D_0$, $D_1$, and $D_2$, and would like to construct trees $T_0$, $T_1$, and $T_2$.
To specify constraints for these matrices, one writes one constraint per line, with first the indices of the distance matrices then the number of rSPR moves allowed between those distance matrices. 
For example, a line saying \texttt{0  2  1} would mean that $T_0$ and $T_2$ are constrained to be one rSPR move apart.
On the other hand, one may specify a linear constraint with a linear constraint parameter. 
If the linear constraint parameter is $L$, then trees $T_i$ and $T_j$ are constrained to be $L \cdot | i - j |$ rSPR moves apart.
So if we apply a linear constraint with parameter 2 in our example, then both $T_0$ and $T_1$ and $T_1$ and $T_2$ are constrained to be at most 2 rSPR moves apart, while $T_0$ and $T_2$ are constrained to be at most 4 rSPR moves apart.

The output for \constNJ\ is collection of correlated sets of trees, each of which get their own \tre\ file, along with a \lengths\ file, which describes the total tree length for each of these sets of trees.
\constNJ\ returns at most one correlated set of trees for each agreement profile within the constraints, which is labeled by the agreement profile.
If the constraints are given in a file, then the agreement profile is written in the order given in the file.
If linear constraints are given, the agreement profile is written as a vector representing an upper triangular matrix in the usual way.
For example, the agreement profile for three trees with linear constraints is written $\left( \drSPR(T_0, T_1), \drSPR(T_0, T_2), \drSPR(T_1, T_2) \right)$, so the set of trees in the file \texttt{example.2\_1\_1.tre} has agreement profile (2,1,1).
The \lengths\ file contains the information on tree lengths, as in Table~\ref{tab:motiv} of the introduction.
Namely, for each correlated set of trees returned by \constNJ, it displays the total tree length for those trees.

\subsection{Speed}

A rigorous worst-case runtime analysis of \constNJ\ would show that it can be incredibly slow.
Indeed, the maximum agreement partition is a generalization of the maximum agreement forest; thus finding the size of the MAP is NP-hard by the corresponding theorem by \citet{bordewichSempleSPR04}.
However, \constNJ\ does not just need to solve one such problem, it needs to solve quite a number of them. 
At worst, \constNJ\ would need to find as many MAP's as there are possible coalescences, just for a single step and a single instance; if an instance had forests with $\ell_1, \dots, \ell_k$ trees, then there will be ${\ell_1 \choose 2} \times \dots \times {\ell_k \choose 2}$ possible coalescences, each of which in theory could require solving of a MAP problem. 
At any step there can be as many instances as there are agreement profiles satisfying the constraint matrices, and a problem with $n$ taxa and $k$ distance matrices will require $n k$ such steps.
Such an analysis would not give a very clear understanding of the practical time requirements of running \constNJ.

In practice, \constNJ\ can be used effectively for a moderate number of taxa and a small number of closely-constrained trees.
The running time depends somewhat on the number of taxa, but quite a lot on the constraints and number of distance matrices.
Indeed, the main bottleneck is the MAP calculation, and the running time of the MAP calculation depends very strongly on the constraints and the number of distance matrices.

However, what may be surprising is how much the running time depends on the quality of the data.
This is vividly illustrated by the simulations, 
where in the case of two trees with two reticulation events and divergence of 0.1 mutation per site per tree, the sequences with 100 sites took on average 10.3 minutes to run, while the simulations with 6400 sites took on average 0.68 seconds each.
This represents a difference of almost three orders of magnitude.
On the same processor (Intel~\textsuperscript{\textregistered}\ Xeon~\textsuperscript{\textregistered}\ CPU at 2.33GHz) using real HIV data, 
an example with three 38-taxon distance matrices and pairwise constraints of three for each pair of distance matrices took 49 seconds, while an example with only two 40-taxon distance matrices with a single constraint of size three took almost 21 minutes.
The quality of the data impacts ``how far'' \constNJ\ has to go down the list of coalescences in order to find one with the desired agreement profile, and how often it needs to calculate a new agreement partition.

We have made some coding choices to increase the speed.
For example, there is a natural partial order on agreement profiles, which is just the element-wise numerical order.
In considering which coalescences to perform, we only investigate those coalescences which could lead to an agreement profile which is smaller than those which have already been performed.
In principle, one could do a more comprehensive search which might lead to more optimal sets of trees; we have not found a significant improvement following such a direction.

\noarxiv{\newpage}
\section{Simulations}

In order to evaluate the performance of \constNJ, we performed a number of simulations.
The trees in the study were generated as follows.
We choose the number of trees in the recombination network, say $k$, the size of the trees, say $n$, and a number of rSPR moves, say $m$. 
We start with a tree $T_1$ drawn from the Yule distribution of trees on $n$ taxa.
After choosing the desired expected number of substitutions on the tree (in simulations below, 0.1, 0.5, and 2), we divided this number by the number of non-root edges in the tree to get the expected number of substitutions per edge.
We then drew the actual number of substitutions per edge from the exponential distribution with the corresponding mean to get the branch lengths of $T_1$.

We then generated $T_{i+1}$ from $T_{i}$ by applying $k$ rSPR moves to $T_i$ as follows. 
For each rSPR move, first select a non-root edge uniformly; call the subtree below the chosen edge $S$. 
Cut off $S$ then glue it back in on a uniformly selected edge of $T_1$ not contained in $S$.
The location along the chosen edge to attach $S$ is chosen uniformly.
Then to simulate differential rates of evolution of different regions, take the average of the previous branch length and a branch length drawn from an exponential distribution as before.

Given such a series of trees $T_1,\dots,T_k$, we generated a collection of distance matrices $D_1,\dots,D_k$ by simulating sequences on the trees.
We did so using the Jukes-Cantor model of sequence evolution with a single rate.
Distances were then calculated using the Jukes-Cantor distance correction \citep{felsensteinBook04}. 
In case the Jukes-Cantor correction gave an undefined value, we repeated the analysis with a new sequence.
We chose the simple Jukes-Cantor model to focus attention on our method rather than the distance estimator.

\begin{figure}[ht]
  \begin{center}
    \hspace{-1cm}
    \arxiv{\includegraphics[width=5in]{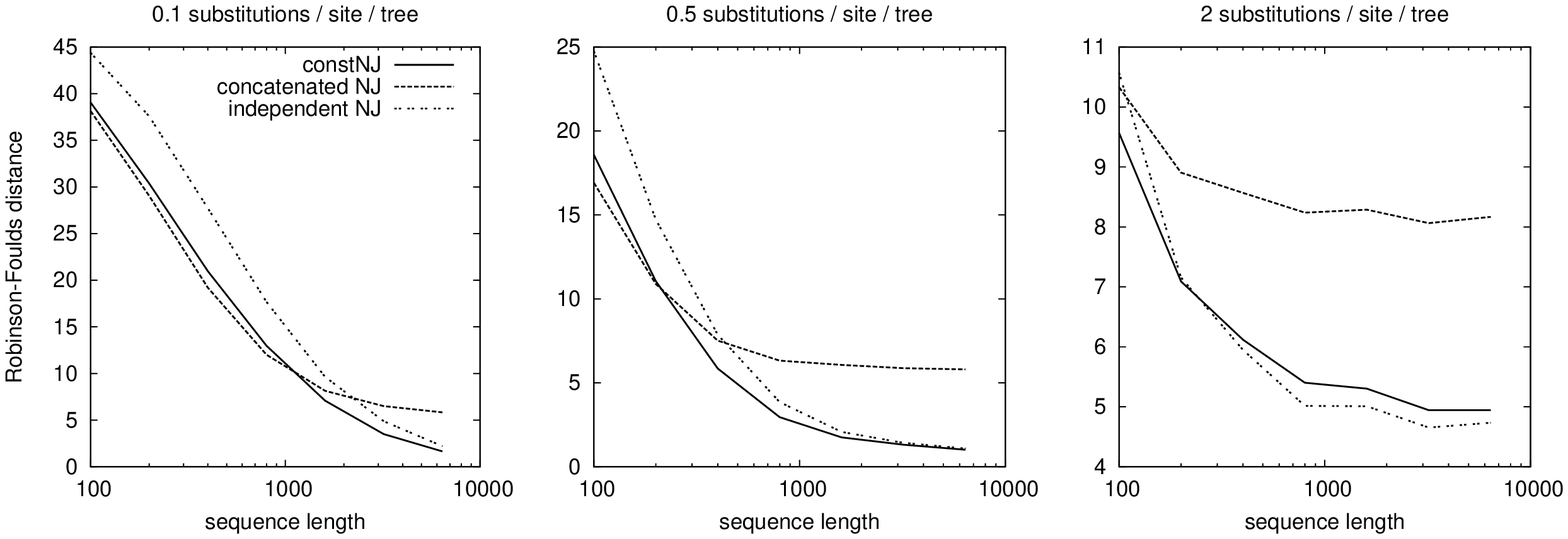}}
  \end{center}
  \caption{
  \constNJ\ simulation results for two trees, each on 30 taxa, averaged over 400 replicates. The first tree was drawn from the Yule distribution, and the second tree was made by applying a random rSPR move to the first. ``constNJ'' is our algorithm, ``concatenated NJ'' is neighbor-joining run with a concatenated alignment, and ``independent NJ'' is neighbor-joining run independently on the alignments for the different trees as described in the text.}
  \label{fig:t2r1n30}
\end{figure}

\begin{figure}[ht]
  \begin{center}
    \hspace{-1cm}
    \arxiv{\includegraphics[width=5in]{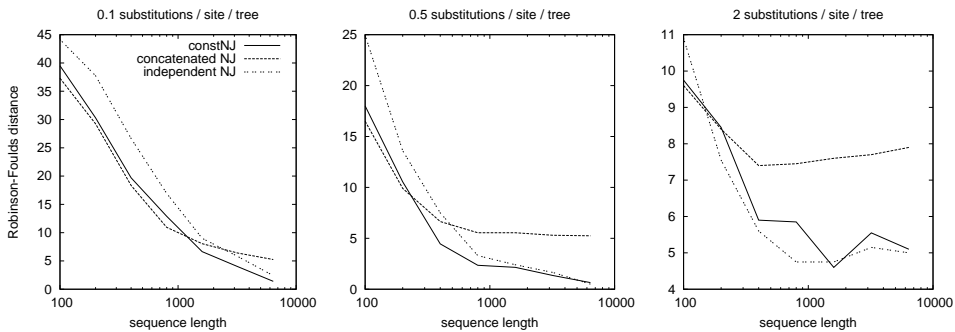}}
  \end{center}
  \caption{
  \constNJ\ simulation results for two trees, each on 30 taxa, averaged over 400 replicates. This time, two rSPR moves were applied to the first tree to get the second. 
  }
  \label{fig:t2r2n30}
\end{figure}

\begin{figure}[ht]
  \begin{center}
    \hspace{-1cm}
    \arxiv{\includegraphics[width=5in]{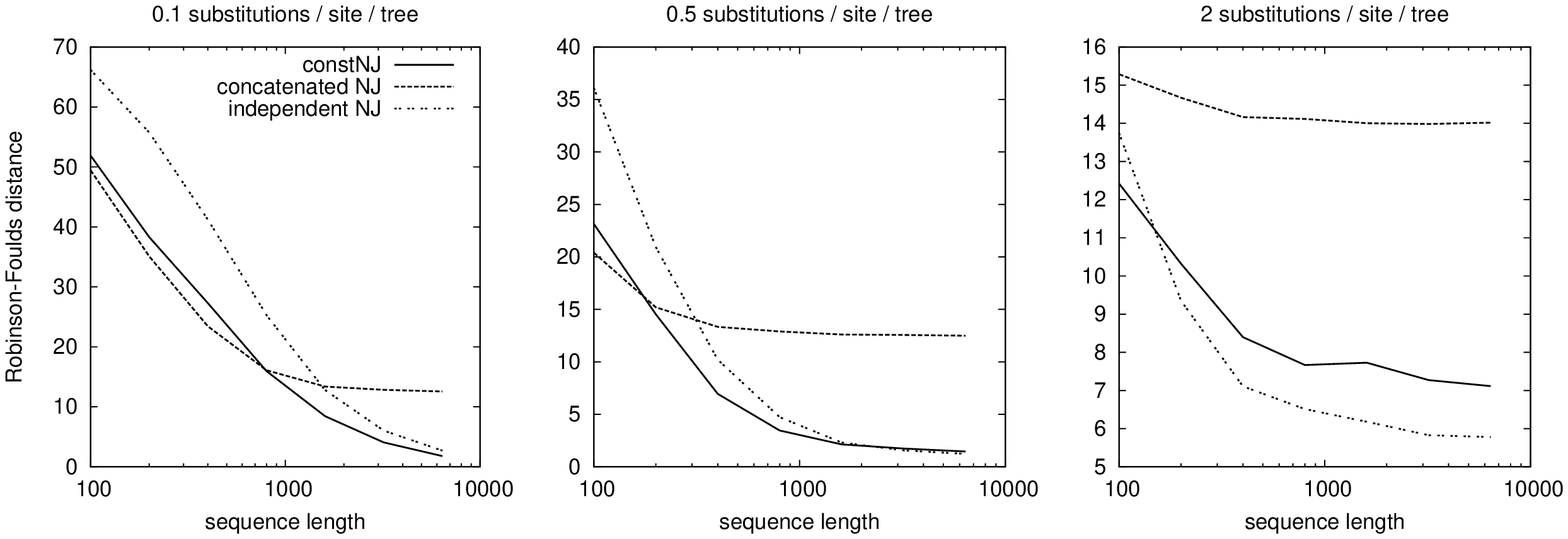}}
  \end{center}
  \caption{
  \constNJ\ simulation results for three trees, each on 30 taxa, averaged over 400 replicates. 
  Here one random rSPR move was done to change the first tree to the second tree, and the second tree to the third. 
  }
  \label{fig:t3r1n30}
\end{figure}

For the first set of simulations, we wanted to understand how the topological accuracy of \constNJ\ compares to that of concatenating alignment blocks or running them independently. 
For concatenation, we estimated a single distance for each pair of taxa by taking the Jukes-Cantor correction of the average number of substitutions in each alignment block; such a procedure simulates the process of concatenating equal-length alignment blocks.
We then considered the resulting tree as the output of running NJ on the concatenated alignment for each alignment block.
For independent construction, we simply ran NJ on each distance matrix independently. 
For \constNJ, we constrained the rSPR distance between the trees to be less than or equal to the number of rSPR moves used to generate the trees. 
The trees used in the comparison were then the shortest (i.e. smallest total tree length) trees returned given those constraints.

To measure topological accuracy, we used the Robinson-Foulds distance \citep{robinsonFouldsComparison81}, which is simply one half the size of the symmetric difference of the edge split sets. 
The results are shown in Figures~\ref{fig:t2r1n30}, \ref{fig:t2r2n30}, and \ref{fig:t3r1n30}. 
In these simulations, \constNJ\ typically outperforms either alternate strategy.
When sequences are short, the main source of error is insufficiently accurate distance estimations;
concatenation increases the amount of useful sequence information for distance estimation, and so outperforms independent construction in that case.
However, \constNJ\ does almost as well.
On the other hand, when sequences are long, independent estimation does well, as there is enough sequence information to reconstruct the tree for each block independently.
In that case, \constNJ\ also does well. 

The reader may object that these graphs represent an unfair comparison, as they assume that the number of reticulation events is correctly bounded in advance. 
The next two simulations address this objection.
The first set, with results shown in Figure~\ref{fig:lengthTest1}, seems to indicate that by looking at the \lengths\ file one can do a reasonable job of deciding how many rSPR moves to allow as was done in the example case of the introduction.
The second set, with results shown in Figure~\ref{fig:interp}, concerns what happens if one makes an incorrect decision.

Figure~\ref{fig:lengthTest1} explores one of the main themes of this paper, which is the trade-off between phylogenetic optimality (in this case total tree length) and congruence among individual trees.
To make this figure, we generated pairs of trees as before, generating a Yule tree and then applying some number of rSPR moves to get the second tree, except that this time we threw out pairs of trees which did not have the correct rSPR distance between them (i.e. when a subtree was moved back to its original location).
We drew branch lengths as above then simulated 1000 sites with an expectation of 0.5 mutations per site per tree.
The $x$-axis is the index of the \lengths\ file, i.e. the number of rSPR moves between the two reconstructed trees.
The left $y$-axis, ``average total length'', shows the average length of trees with that number of rSPR moves between them.
For instance, consider the point on the line labeled ``three rSPR moves'' which is at $x$-value 2. 
This says that if we simulate a pair of trees which are three rSPR moves apart as described above, then we expect the pair of trees output by \constNJ\ with agreement profile two to have total length about 1.038.
Note that \constNJ\ does not always return a tree for every agreement profile which is allowed under the constraints.
In those cases, we simply took the total tree length from the largest non-empty agreement profile.

\begin{figure}
  \begin{center}
    \arxiv{\includegraphics[width=3in]{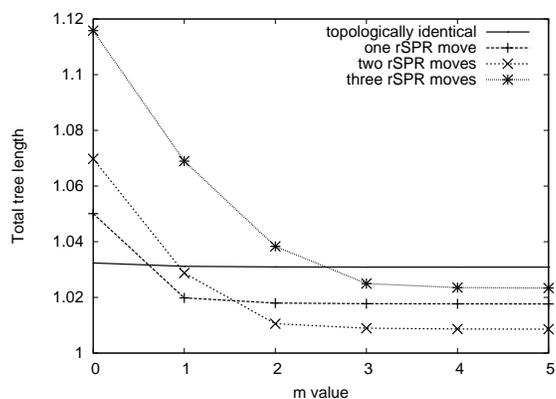}}
  \end{center}
  \caption{
  Comparison of the total tree lengths for simulated trees differing by the described number of moves and then reconstructed using \constNJ; average of 100 simulations.
  As can be seen, the most significant decreases in the total tree length happen when getting to the correct number of rSPR moves, after which the plot levels off. 
  For example, on the line ``two rSPR moves,'' there are significant decreases in length when going to one and two rSPR moves, but not much decrease after that.
  Thus, at least in simulation, it appears possible to make a reasonable choice concerning the number of rSPR moves to allow between the two trees.
  }
  \label{fig:lengthTest1}
\end{figure}

Figure~\ref{fig:lengthTest1} shows exactly what one might expect. 
Namely, if we generate pairs of identical trees, then not much improvement in terms of total tree length is gained by allowing the trees to differ.
However, if the trees are one rSPR move apart, then there is a substantial drop when allowing one rSPR move, but not much more after that; this indicates that only one rSPR move is called for by the data.
The situation is similar for the other numbers of rSPR moves.
Thus, at least in simulation with good quality data, it appears that one should be able to make a reasonable judgment as to the correct number of rSPR moves for the data set at hand, as was done in the introduction.

We also performed some simulations allowing an incorrect number of rSPR moves (Figure~\ref{fig:interp}).
As shown there, giving a too-small constraint interpolates between results from concatenated data and the correct specification, while too-large constraints give performance similar to the correct constraint.
One might expect \constNJ\ with too-large constraints to give results similar to independent NJ; we do not have a clear explanation why this is not the case.

\begin{figure}[ht]
  \begin{center}
    \hspace{-1cm}
    \arxiv{\includegraphics[width=3in]{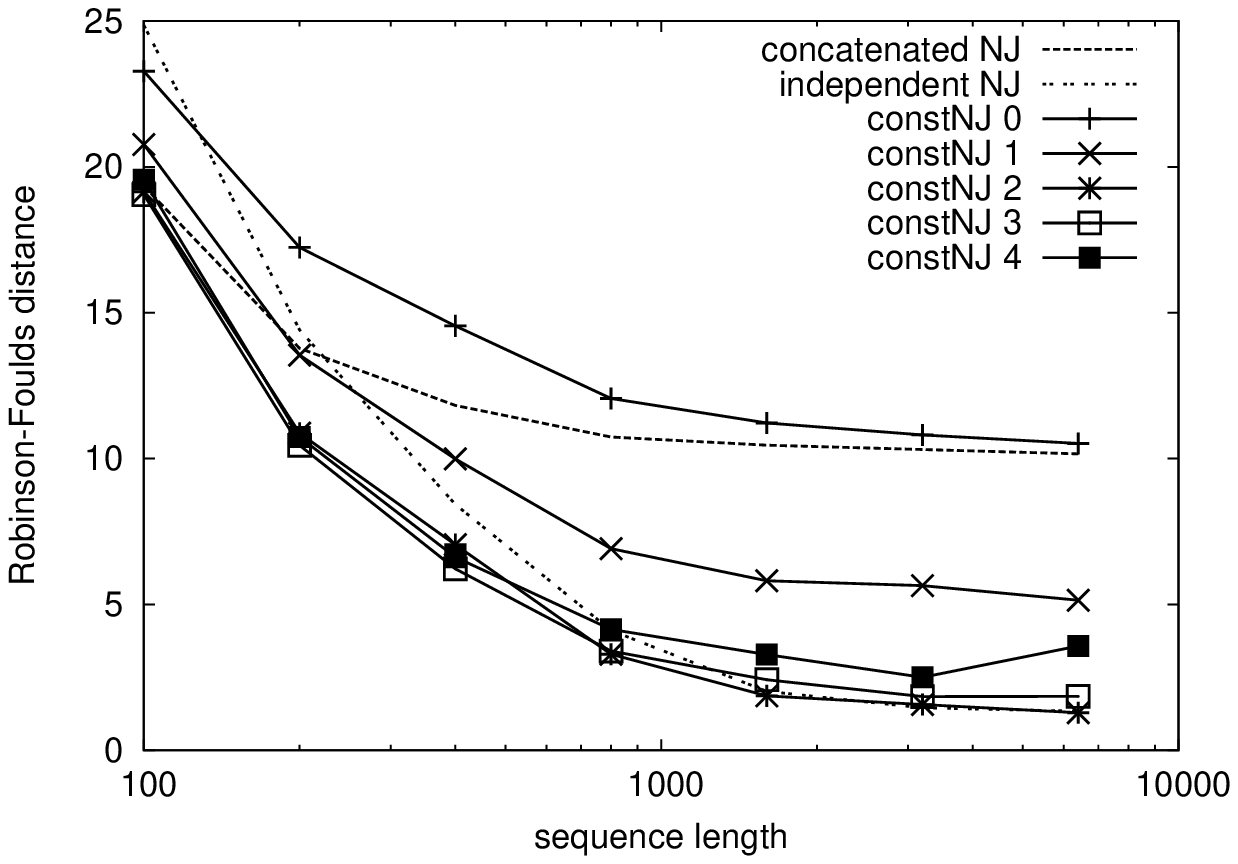}}
  \end{center}
  \caption{
  Comparison of various specified constraints for \constNJ; average of 100 simulations. Data was simulated on two trees, each on 30 taxa, such that two random rSPR moves were done to change the first tree to the second tree. Then reconstruction was done with rSPR distance constraints of 0, 1, 2, 3, and 4. As would be expected, having a constraint of 0 (identical trees) has qualitative performance similar to that of concatenated NJ, the best performance is obtained by the correct constraint of 2 moves, and a constraint of 1 gives results between those for 0 and 2. The performance of 3 is similar to that for 2, while 4's performance degrades with accurate distances.
  }
  \label{fig:interp}
\end{figure}

\noarxiv{\newpage}
\section{Conclusion}

In this paper we present \constNJ, a consistent distance-based algorithm for a collection of trees with pairwise rSPR constraints, such as those constraints satisfied by collections of trees which fit into a reticulation network. 
\constNJ\ is deterministic and a strict generalization of the neighbor-joining algorithm. 
In order to ensure that the resulting set of trees satisfy the specified constraints on rSPR distance, we develop the theory of maximum agreement partitions, culminating in Theorem~\ref{thm:main}.

We hope that this algorithm is the beginning of a new direction for phylogenetic inference of reticulation networks.
We simplify the problem considerably by assuming that the alignment blocks are known in advance; in doing so we preserve the correlation between sites in the alignment with the same history.
Rather than first finding trees and then attempting to put them into a recombination network, we investigate the balance between discordance between trees and optimality of the ensemble of trees.
By using trees whose rooting is derived from outgroups, we find explicit evolutionary histories.

To our knowledge, \constNJ\ is the first algorithm of its kind, and there is considerable room for improvement over this first attempt.
First, we enforce pairwise bounds on the rSPR distance between trees, which is a relatively weak way to show that these trees fit into a network. 
A more explicit approach would be desirable. 
Second, distance-based methods waste a considerable amount of information which is used by likelihood-based methods; a logical next step would be to create a likelihood-based method. 
Doing so would require a collection of ``moves'' analogous to rooted nearest-neighbor-interchange or rSPR moves for a heuristic search, but for collections of trees, with the constraint that the moves don't change the rSPR distance between trees too much.
One option would be to explicitly store a reticulate network in memory and have the trees moving about inside the reticulate network while the network changes.
Such an algorithm would do a better job of actually reconstructing a reticulate network.
Third, an alternative direction for heuristic optimization might be to do a more complete search for the minimum length tree in a manner analogous to algorithms searching for the BME tree.
Fourth, a more immediate issue is that \constNJ\ does not reconstruct branch lengths. 
It would be possible to do a distance-based branch length estimation in a manner similar to that for usual neighbor-joining, but the fact that we are choosing trees which may be sub-optimal according to the NJ criterion implies that negative branch lengths might be encountered.
Alternatively, one might use a program such as PHYML \citep{guindonGascuelPhyml03} to estimate branch lengths on each fixed topology independently.
This appears to be a reasonable way to proceed, despite the fact that the correlation between branch lengths of the different trees is lost.
A more correct approach will require some sort of correlation of the branch lengths in a model-based manner, and we believe that such reconstruction is probably best done in the context of a complete likelihood-based approach as described above.
Finally, it might be interesting to incorporate some \constNJ\ ideas into bootscanning-type methods for recombination breakpoint inference.

Eight years ago, \citet{kuhnerEaMlRecombination00} wrote
``[w]hen recombination occurs adjacent sites may have different, although correlated, genealogical histories. Reconstructing these genealogies with certainty is impossible.'' 
Although we do not claim certainty for this (or any forthcoming) algorithm attempting to reconstruct reticulate phylogenetic history, we think that there is cause for optimism and look forward to seeing future developments in this area.

\subsection*{Acknowledgments}

The author would like to thank Lior Pachter for a number of helpful early discussions.
He is also very grateful to the Mullins HIV lab at the University of Washington for the ongoing collaboration which led indirectly to the present work, and to the Armbrust Oceanography lab at the University of Washington for use of their computing cluster.
The author made use following useful software: Figtree \citep{rambautFigtree}, PHYML \citep{guindonGascuelPhyml03}, and PHYLIP \citep{felsensteinPhylip04}.

\bibliographystyle{jcbnat}
\bibliography{distrec}

\end{document}